\documentclass[12pt,titlepage,a4paper]{article}

        \usepackage{mathtools,amsthm,amssymb}

        \usepackage{natbib}
 
        \usepackage{hyperref}

        \hypersetup{colorlinks=true,urlcolor=blue,citecolor=blue,linkcolor=green}

        \usepackage[center]{caption}
	
	\usepackage{setspace}
 	
	\usepackage{epsfig}
	
	\usepackage{graphics}
	
	\usepackage{lscape}

	\usepackage[english]{babel}
	
	\usepackage{color}

   \newcommand{\one}{\mathbf 1}
   \newcommand{\Rea}{\mathbb R}

        \theoremstyle{plain}

        \theoremstyle{plain}

        \theoremstyle{plain}\newtheorem {lemma}{Lemma}

        \theoremstyle{plain}\newtheorem {proposition}{Proposition}

        \theoremstyle{plain}\newtheorem {theorem}{Theorem}

        \theoremstyle{definition}\newtheorem {assumption}{Assumption}

        \theoremstyle{definition}\newtheorem{definition}{Definition}

        \theoremstyle{remark} 

        \theoremstyle{remark} 

        \theoremstyle{remark}

\begin{document}

  	\bibliographystyle{plainnat}
      
   	\bibpunct{(}{)}{;}{a}{,}{,}
        

        \title{Exact inference from finite market data} 

        \author{F. K\"ubler \thanks{University of Z\"urich \href{mailto:fkubler@gmail.com} {\tt fkubler@gmail.com}} \and 
                     R. Malhotra \thanks{University of Warwick \href{mailto:r.malhotra@warwick.ac.uk} {\tt r.malhotra@warwick.ac.uk}}  \and 
                     H. Polemarchakis \thanks{University of Warwick \href{mailto:h.polemarchakis@warwick.ac.uk} {\tt h.polemarchakis@warwick.ac.uk}}}

        \date{July 14, 2021}

        \maketitle

        \begin{abstract}
        
        We develop conditions under which individual choices and Walrasian equilibrium prices and allocations can be exactly inferred from finite market data. First, we consider market data that consist of individual demands  as prices and incomes change. Second, we show that finitely many observations of individual endowments and associated  Walrasian equilibrium prices, and only prices, suffice to identify individual demands and, as a consequence, equilibrium comparative statics.

        	\bigskip 
        
       	\bigskip


        \bigskip


        	\end{abstract}

 	From only a finite number of observations of market data,  can we infer how the individual shall decide when faced with choices not previously encountered? Can we infer equilibrium prices and allocations as the distribution of endowments varies? Do market data need to consist of individual demands or do observations of aggregate demand or of the equilibrium price correspondence suffice for identification?
	
	Variations of this question have been extensively studied, but most of the existing literature either focuses on the case of infinitely many observations or poses only the question whether observations are consistent with utility maximization. The seemingly very important question of what one can conclude about an individual's preferences from finitely many observations has largely been overlooked. In this paper, we  pose exactly this question: what does one need to know about preferences a priori in order to be able to make non-trivial exact inference about the underlying data generating preference with only finitely many data points? In particular can one predict how an individual will choose from a choice-set not preciously encountered? Can one predict equilibrium prices at a profile of individual endowments not previously encountered?
		
	Revealed preference analysis, the weak axiom, was introduced by \citet*{samuelson1938} as a necessary condition for demand data, a collection of pairs of prices and bundles of commodities, to be generated by the maximization of a preference relation subject to the budget constraint. \citet*{houthakker1950} introduced the strong axiom, sufficient for demand data to be generated by preference optimization. Later, \citet*{afriat67} established the generalized axiom of revealed preference as necessary and sufficient for a finite set of demand data to be derived from the maximization of a preference relation or ordinal utility function. \citet*{reny15} extended the argument to arbitrary data sets. \cite{varian82} did use the results in \cite{afriat67} to discuss the possibility of making statements about preferences from a finite number of observations, but, he gave a very partial answer to the problem.
	
	\citet*{hurwiczuzawa71} and \citet*{mascolell77-2} gave necessary and  sufficient conditions for  the integrability of a demand function, the derivation of  the generating ordinal utility function,  and for a demand function to identify preferences. It is worth noting that, though Lipschitz continuity of preferences or, alternatively, Lipschitz continuity of income expansion paths is sufficient for identification, conditions on preferences that guarantee that demand is Lipschizian in income are not known. More importantly, though integrability answers our motivating question affirmatively for an infinite number of observations, it says nothing for the case of a finite number of observations. It does not even address the question of asymptotics.
	
	\citet*{mascolell78} gave sufficient conditions to ensure that, for a nested, increasing sequence of demand data that, at the limit, cover a dense subset of consumption choices, any associated sequence of preference relations converges to the unique preference relation that generated demand. In a recent paper, \citet*{chambersetal19} considered the case of pairwise choice and showed that convergence fails. When data sets that are collections of choices from pairwise comparisons of alternatives become dense, generating preferences may not convergence to the unique underlying preference even if the underlying preference relation is continuous. Convergence obtains if the data satisfies a condition implied by, but weaker than monotonicity. For choices generated by a monotone preference relation, convergence follows from \citet*{forgesminelli06}.
	
	\cite{brownmatzkin96} extended \cite{afriat67} to a framework where one has a finite number of observations on profiles of equilibrium endowments and Walrasian equilibrium prices. That is, observations on the equilibrium manifold. \cite{chiapporietal04} showed that the equilibrium manifold locally identifies individual preferences. \cite{brownmatzkin90} and in particular \cite{matzkin06} proved a global version of the result.

	In this paper, we pose the question of inference from a finite number of observations. We consider two classes of finite data. First, we examine the case where one has observations on an individual's demand at different prices and incomes. Second,  for the case of an exchange economy, we assume that only equilibrium prices and profiles of individual endowments are observable. In both settings, we consider sequences of nested observations that become dense in the chosen domains as in \cite{mascolell78}\footnote{Equivalently one can assume  
	that observations are drawn randomly from a uniform distribution of exogenous variables and that the endogenous data is generated by utility maximizing individuals. However, in this case all our statements only hold with probability one.}
	In the case of demand, we allow for all positive prices and incomes. In the case of equilibrium, we consider all aggregate endowments and any compact set of income distributions.
	
	For both setups, we show that, given any two consumption bundles between which an individual is not indifferent, after sufficiently many observations, we can infer the individual's choices. We first illustrate this in a preliminary example below. Instead of considering market data, we assume, there, that we observe an individual's choices among pairwise alternatives. In that context, it is very easy to show that, assuming monotonicity and continuity  of preferences, the result must hold true. For the more complicated case of individual demand, this fact follows from a result by \cite{mascolell77-2}, who shows that,  if a consumption bundle $x$ is preferred to a bundle $ y, $ then there must exist prices and incomes such that $x$ is revealed preferred to $y$. We show that, if $x$ is strictly preferred to $y,$ and, if prices and incomes are randomly drawn, then eventually (after finitely many draws) we must observes prices and incomes so that $ x $ is strictly revealed preferred to $y$. From this result, we deduce that, for any price vector $p$ and income $\tau, $ and any $ \epsilon > 0, $ there are finitely many observations that allow us to determine individual demand at $ (p,\tau) $ within $ \epsilon $. Our main result is that the situation is identical for the case where only equilibrium prices are observable. If $ x $ is strictly preferred to $y$ by some individual, then sufficiently many observations of equilibrium  price and endowment profiles suffice to prove that the individual strictly revealed prefers to $ x $ to $y$.
	
	A question that arises directly from our analysis is whether any of these results extend to equilibrium  comparative statics. The transfer paradox, introduced by \cite{leontief1936-1}, makes it clear that knowledge of utility functions is necessary in order to identify even the direction of welfare effects of transfers. We show that, while predictions of exact comparative statics are generally impossible, we can predict approximate equilibrium prices from finite data. This is true for both setups, independently of whether observations consist of individual choices as prices and incomes vary or of equilibrium prices as the profile of individual endowments varies.
	
	The rest of the paper is organized as follows. Section 1 lays out some preliminaries and provides a simple example that illustrates the main idea of the paper. Section 2 considers the case where individual demand is observable, Section 3 focuses on equilibrium.

	\section{Preliminaries}
	\label{sec:1}
	
	We collect definitions and results needed for our analysis.

	\subsection{Preferences}

	An individual has a preference relation $ \succeq $ over $ {\mathbf X} \subset \Rea^L_+. $ The preference relation $ \succeq $ is upper semi-continuous if, for every $ x\in {\mathbf X},$ the upper contour set, $ {\mathbf R}_+(x)=\{ y : y \succeq x \}, $ is closed. It is continuous  if, for every $ x\in {\mathbf X},$ the upper contour set as well as the lower contour set, $ {\mathbf R}_-(x)=\{ y : x \succeq y \}, $ is closed. It is monotonically increasing  (or simply monotone) if $ x >y $ implies $ x \succ y.$ It is convex if $ x \succeq y $ implies that $ \lambda x + (1-\lambda) y \succeq y $ for all $ \lambda \in [0,1], $ and it is strictly convex if $ x \succeq y $, $ x \ne y $ implies that $ \lambda x + (1-\lambda) y \succ y $ for all $ \lambda \in (0,1). $ 
	
	Following \cite{mascolell77-2}, Remark 4, continuous, monotone and convex preferences $ \succeq $ are Lipschitzian if, for every $ r>0, $ there are numbers, $ H>0 $ and $ \epsilon > 0, $ such that, if 
	\[
	x,y,z \in {\mathbf X}_r = \{ w \in  {\mathbf X}: \frac{1}{1+r} \one \le w \le (1+r) \one 	\}, 
	\] 
$ x \sim y $ and $ \| x-z \| < \epsilon $, then 
	\[
	\delta(x,{\mathbf R}_+(z)) \le H \delta(y,{\mathbf R}_+(z)),
	\] 
where, for a set $ {\mathbf A} \subset {\mathbf X} $, $ \delta(x,{\mathbf A}) = \inf_{y \in {\mathbf A}} \| x-y \|.$

	\subsection{Demand}

	Budget sets are 
	\[
	{\mathbf B}(p)=\{ x \in {\mathbf X} : p \cdot x \le 1 \}, \quad p \gg 0,
	\]
and the Walrasian demand correspondence is defined by  
	\[
f^{\succeq}(p) =\{ x \in {\mathbf B}(p): x \succeq y, \mbox{ for all } y\in {\mathbf B}(p)\}. 
	\] 

	For a monotone, strictly convex and continuous  preference relation $ {\succeq}, $ we use the same notation, $ f^{\succeq}(\cdot) $, to denote the (continuous) Walrasian demand function that is generated by $ {\succeq} $.

	\cite{mascolell77-2} shows that, if preferences are, in addition, Lipschitzian, then
	\[
	\succeq = \succeq' \quad \Leftrightarrow \quad                                                                    	f^{\succeq}(\cdot) = f^{\succeq'}(\cdot). 
	\]
This result we refer to as identification. As will become clear below, it is a necessary condition for most of our results on inference from finite data.

	\subsection{Equilibrium}

	We consider a pure exchange economy with $H$ individuals, $ h \in {\mathbf H}, $ and $L$ commodities. Consumption sets are $ {\mathbf X}=\Rea^L_{++}, $ and each individual, $h$ has continuous, monotone and strictly convex preferences over $ {\mathbf X}, $ that we denote by $ \succeq^h, $ and endowment $ e^h \in \Rea^L_{++}. $ A profile of preferences across individuals is $ \succeq^{\mathbf H}, $  and a profile of endowments across individuals is $ e^{\mathbf H} \in \Rea^{HL}_{++} $. The equilibrium correspondence is defined as 
	\[
	{\mathbf P}^{\succeq^{\mathbf H}} (e^{\mathbf H})=\{ p \in \Delta^{L-1} : \sum (f^{\succeq^h}(\frac{p}{p \cdot e^h}) - e^h)=0 \} .
	\] 
\cite{chiapporietal04} and \cite{matzkin06} derive sufficient conditions on preferences that ensure classical identification in the sense that, for any two profiles of preferences $ \succeq^{\mathbf H} $ and $ \widetilde{\succeq}^{\mathbf H}, $
 	\[
	\succeq^{\mathbf H} = \widetilde{\succeq}^{\mathbf H} \Leftrightarrow
   {\mathbf P}^{\succeq^{\mathbf H}} (\cdot)={\mathbf P}^{\widetilde{\succeq}^{\mathbf H}} (\cdot) .
   \]
As \cite{balasko04} points out, without restrictions on the domain of the equilibrium correspondence, this problem becomes trivial since it reduces to the individual problem when individual endowments are on the boundary for all but one individual.
   
   	More interestingly, \cite{chiapporietal04} and \cite{matzkin06} provide a local version of the result. For our setting of finitely many observations, to state the local result it is useful to observe that this problem is identical to the problem of classical identification of aggregate demand.
	
	For individual incomes $ (w^1,\ldots,w^H) \in \Rea^H_{++}, $ we define the
aggregate demand function
	\[
	d^{\succeq ^{\mathbf H}}(p,w^1,\ldots, w^H) = \sum_{h \in {\mathbf H}} 			f^{\succeq^h}	(\frac{p}{w^h}) .
	\]

	Note that, for any $ e \in \Rea^L_{++} $ and any profile of incomes $ (w^1,\ldots, w^H),$
	\[
	p \in {\mathbf P}^{\succeq^{\mathbf H}}(\frac{w^1}{\sum_{h \in {\mathbf H}} w^h} e, \ldots, \frac{w^1}		{\sum_{h \in {\mathbf H}} w^h} e)  \Leftrightarrow d(p,w^{\mathbf H})= e.
	\]
Therefore, for two profiles of preference relations $ \succeq^{\mathbf H} $ and $\widetilde{\succeq}^{\mathbf H}, $
	\[
	d^{\succeq ^{\mathbf H}}(\cdot) =d^{\widetilde{\succeq} ^{\mathbf H}}(\cdot) 
	\Leftrightarrow   {\mathbf P}^{\succeq^{\mathbf H}} (\cdot)={\mathbf P}^{\widetilde{\succeq}	^{\mathbf H}} (\cdot).  
	\]
To state the local version of the result, we define 
	\[
	{\mathbf W} = \times_{h\in {\mathbf H}} [\underline{w}_h,\overline{w}^h],
	\]
for some bounds $ 0 < \underline{w}_h < \overline{w}_h $, $ h \in {\mathbf H} $.
While prices are observed globally, it suffices to observe incomes locally. As \cite{matzkin06} points out, this differentiates crucially this setting from \cite{balasko04}, whose argument relies on allowing incomes to be at the boundary.
We make the following high level assumption on preferences:

	\begin{assumption} \label{ass1}

	For all $ \succeq^{\mathbf H} \neq \succeq^{ \prime\mathbf H},$ there exists a $ \bar p \in \Rea^{L}_{++} $ as well as incomes $ \bar w^{\mathbf H} \in {\mathbf W}, $
such that
	\[
	d^{\succeq^{\mathbf H}}(\bar p, \bar w^{\mathbf H})
	\neq d^{\succeq^{\prime\mathbf H}}(\bar p, \bar w^{\mathbf H}). 
	\]

	\end{assumption}

	By the continuity of the aggregate demand function, the assumption implies that there must exist an open neighborhood of prices and incomes where aggregate demand differs.
\cite{matzkin06} gives assumptions on individual demand functions that are sufficient for \ref{ass1} to hold. \cite{chiapporietal04} derive sufficient conditions on preferences that ensure the assumption. These are conditions on the income effects of the individual demand functions. To state the assumptions, let $ v^h_l $ denote the derivative of demand of individual $h$ for commodity $l$ with respect to income. We require that, for every individual,
	
	\begin{enumerate}
        
        \item   for every commodity, the income effect $ v^h_l(\cdot) $ is a twice differentiable function of income, $w$ and
        \[ 
        \frac{\partial v^h_l}{\partial w} \neq 0,
        \]
        
        while

        \item there exist commodities, $ m \neq 1 $ and $ n \neq 1 $ , such that \[ \frac{\partial}{\partial w}(\ln \frac{\partial v^h_m}{\partial w}) \neq \frac{\partial}{\partial w}(\ln \frac{\partial v^h_n}{\partial w}). \]

        \end{enumerate}
	
	These two assumptions from \cite{chiapporietal04} ensure that income effects  do not vanish for any commodity, while there are two commodities for which the partial elasticities of the income effects with respect to revenue do not vanish. The assumptions only apply to the case of at least three commodities.

	\subsection{A preliminary result}

	In order to motivate our results it is useful to illustrate the main idea in a setting that is much simpler than the demand setting. For this, we consider the simple binary choice problem. That is, we observe a sequence of choice sets
   	\[
	({\mathbf A}_k=\{x_k,y_k \}, \quad \  x_k,y_k \in {\mathbf X},  \quad k = 1, \ldots)
	\] 
and the associated choices by an individual who has a preference relation over $ {\mathbf X}. $ We assume that, as $ n \rightarrow \infty, $ the set $ \cup_{i=1}^n {\mathbf A}_i $ becomes dense in $ {\mathbf X} \times {\mathbf X}.$ Given observations $ {\mathbf A}_k $, $ k=1,\ldots, n $ we say that $ x $ is strictly revealed preferred to $y$, $ x \succ^{R_n} y $,
if, for every preference relation  that is consistent with the $n$ observed choices, $x$ is strictly preferred to $y$.  

	In the absence of restrictions on preferences relations, this is a meaningless definition. However,  it is easy to see that the assumption of monotone preferences does allow for meaningful inference.
	
	If an individual chooses a commodity bundle $ x $ over another bundle $y$, as in Figure \ref{fig:1}, it is revealed that she must prefer any bundle in the set $ {\mathbf X}=\{ x' \gg x \} $ to any bundle in the set $ {\mathbf Y}=\{ y' \ll y \} $.

	\begin{figure}[tb]
	
	\begin{center}
		\includegraphics[scale=0.25]{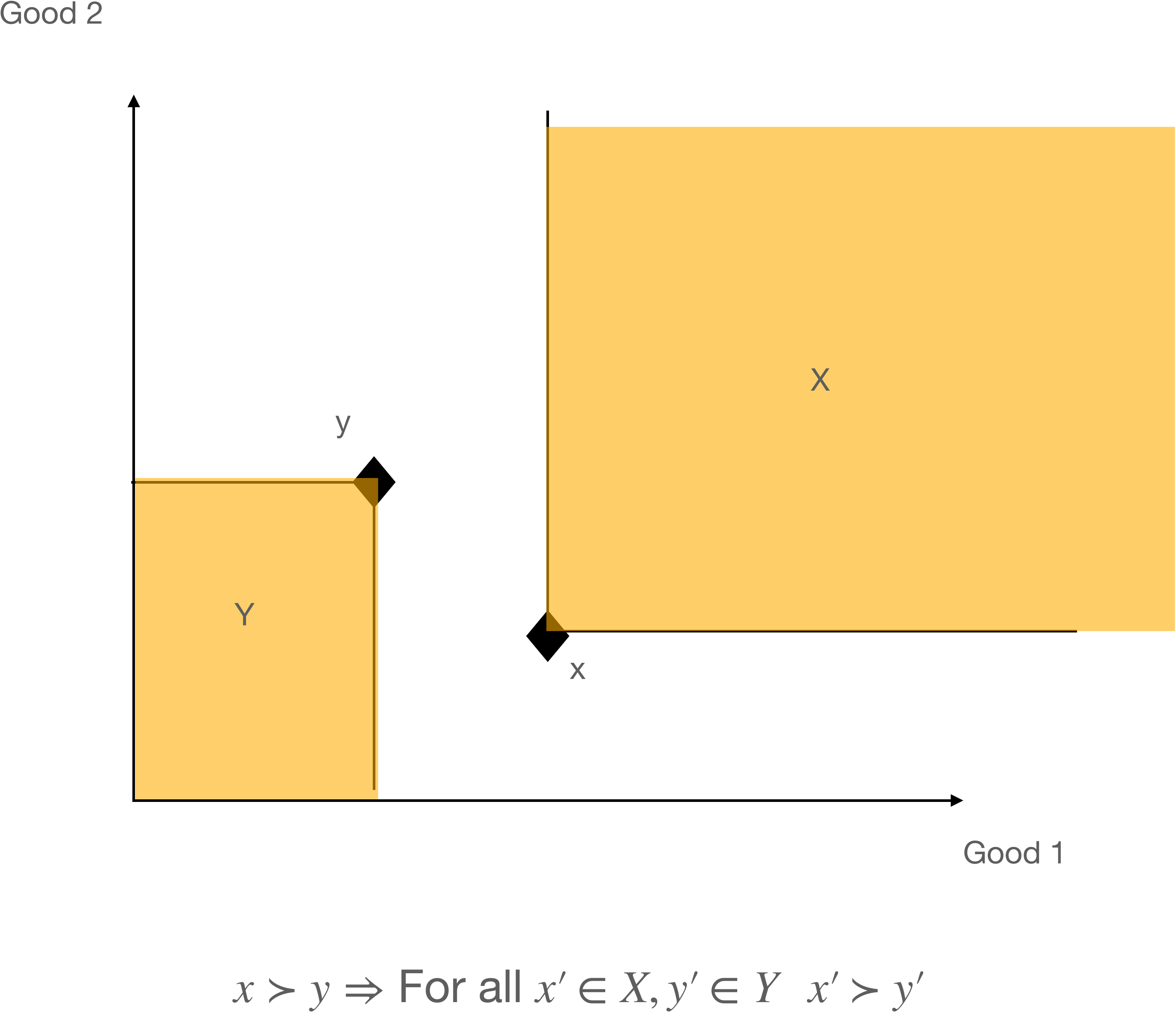}\\
	\end{center}
	\vspace{-0.5cm}
	\footnotesize{}.
	\caption{Revealed preference from pairwise choice}
	\label{fig:1}
	
	\end{figure}
		
	It is then straightforward to show the following result:
  	
  	\begin{proposition} \label{prop1}
  	
  	Suppose preferences are monotone and continuous. For any  $ x, y \subset int({\mathbf X}) $ with $ x \succ y, $ there is an $n$ such that $ x \succ^{R_n} y $.

	\end{proposition}

	\begin{proof}

	Given and $x,y $ with  $ x \succ y $, as $ \cup_{i=1}^n {\mathbf A}_i $ becomes dense, by continuity, there must exist a $k$ with  $ {\mathbf A}_k=\{ \bar x, \bar y \},$ such that $\bar x < x $ and $ \bar y > y, $ and
		\[
		f( {\mathbf A}_k )= \bar x .
		\] 
By monotonicity, any $ {\succeq}'$ that rationalizes $ ({\cal A}_k,f) $ must satisfy $ x {\succ'} y$. 
	
	\end{proof}

	\bigskip

	The result can also be inferred from \cite{chambersetal19}. The proof  here  is much simpler, but it does not show identification over all of $ {\mathbf X} $ in the limit. 
	
	The point is that the assumption of monotonicity of preferences allows us to make statements about arbitrary $x,y$ which are not previously observed.

	\section{Individual demand}
	\label{sec:2}

	Throughout this section, we assume that we observe market choices of a single individual who has preferences $ \succeq $ over $ {\mathbf X} = \Rea^L_{+} $. We consider a possibly infinite sequence or prices $ (p_k: k=1,\ldots) $ that become dense in $ \Rea^L_{++} $ as $k \rightarrow \infty $. We hold the sequence fixed throughout the argument, and we assume  that $n$ observations consist of the first $n$ prices of this sequence together with optimal choices, $ (p_k, x_k), k=1,\ldots, n, $ where
 $ x_k = f^{\succeq} (p_k), $ for all $ k=1,\ldots, n.$ With a slight abuse of language, we sometimes refer to this as a sequence of random observations. This is to emphasize that we impose no restrictions on the prices except that they become dense in the limit.
 
 	The two questions we pose are as follows:

	\begin{enumerate}
    
    	\item Given arbitrary $ x,y \in {\mathbf X}, $ with $ x \not\sim y,$ can we determine how  the individual shall choose from the set  $ \{x, y\}$ after observing some number $n,$ of market choices?
    
    	\item Given arbitrary prices $ p \in \Rea^L_{++}, $ can we predict the individual's demand at these prices from some number $n,$ of observations of choices?

	\end{enumerate}

	The answers to both questions turn out to be ``yes." But, it is important to point out that posing the questions slightly differently leads to the opposite conclusion: Fixing any number of observations, $n$, there obviously always exist $x,y \in {\mathbf X}$ for which one cannot determine the individual's preference.

	It is well known that a finite number of observations can be rationalized by any strictly convex, continuous and monotone preference relation if and only if they satisfy the strong axiom of revealed preferences. For completeness it is useful to state the strong axiom.
 
	\begin{definition}

	Observations satisfy the strong axiom of revealed preferences (SARP) if for every ordered subset $ \{i_1,i_2,...,i_m\} \subset {\mathbb N} $ with
$ x_{i_k} \ne x_{i_j} $ for all $k,j,$ and with
	\[
	\begin{array}{ccc}
	p_{i_1} \cdot x_{i_2} & \le & p_{i_1} \cdot x_{i_1},\\
	p_{i_2} \cdot x_{i_3} & \le & p_{i_2} \cdot x_{i_2},\\
 	& \vdots & \\
	\end{array}
	\]
it must be the case that 
	\[
	p_{i_m} \cdot x_{i_1} >  p_{i_m} \cdot x_{i_m}. 
	\]

	\end{definition}

	\subsection{Identification of pairwise choices}

	Reconstruction of partial preferences from a finite number of observations on prices and choices is an obvious application of revealed preference analysis. \cite{varian82} made the point. Figure \ref{fig:2} illustrates the basic idea. If we observe that a bundle $x$ is chosen at some prices $p$ and a bundle $y$ is chosen at prices $q$, and if $y$ lies below the budget line of $x$ we can infer that $x$ is strictly preferred to $y$ and that all bundles strictly greater than $x$ are strictly preferred to all bundles in the budget set at prices $q$.
  	
	\begin{figure}[tb]
	
	\begin{center}
		\includegraphics[scale=0.25]{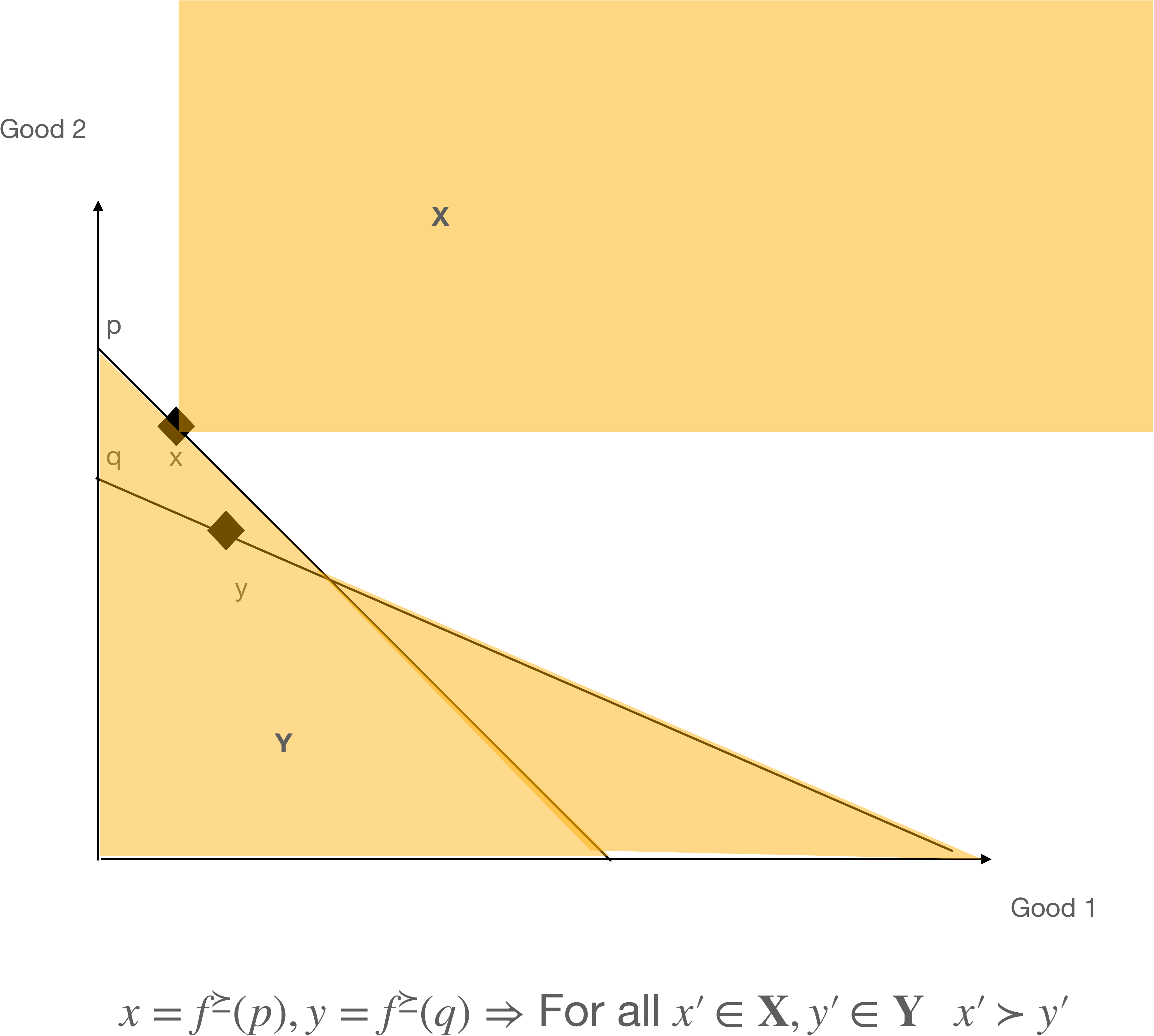}\\
	\end{center}
	\vspace{-0.5cm}
	\footnotesize{}.
	\caption{Revealed preference from demand}
	\label{fig:2}
	
	\end{figure}

  However, the question whether, given a sequence of random observations and any $x \succ y $ there must be some finite number of observations after which $ x $ is revealed preferred to $y$ is not addressed in that literature.
 
 	For $x,y \in int({\mathbf X}),$ we say that $x$ is strictly revealed preferred to $y$ through $n$ observations, $ x \succ^{R_n} y $,  if there are $k$ observations indexed by $ i_1, \ldots, i_k \in \{1,\ldots,n \}$, such that $x \ge x_{i_1}$, $ p_{i_j} \cdot x_{i_{j+1}} < p_{i_j} \cdot  x_{i_j} $ for all $ j=1,\ldots k $ and {\color{green}{$ x_{i_N}\ge  y $}}. Evidently,  if $ x \succ^{R_n} y, $ for some $n,$ then $ x \succ y $. Our first result is a converse: if $ x \succ y, $ then there is an $n$ such that $ x \succ^{R_n} y $.

 	The result builds on \cite{mascolell77-2}, Remark 12 and, in parts, our proof closely follows the argument there. The idea of the proof is to show that for $ y \in {\mathbf X} $ the set of $ z \in {\mathbf X} $ that are not strictly revealed preferred to $y$ for any $n$ (call this set $ T_y $) is identical to the upper contour set of $y:$ the set of commodity bundles that are weakly preferred to $y$. In order to do so, we follow \cite{mascolell77-2} and construct a new preferences relation $ \succeq '$ that is identical to $ \succeq $ precisely when $ T_y $ is identical to the upper contour set at $y$. The key is to show that $ \succeq' $ generates the same demand function as $ \succeq $. Since $ \succeq $ is Lipschitzian, this implies that the two preference relations must be identical.

	\begin{theorem} \label{thm1}

	If $ {\succeq} $ is continuous, Lipschitzian, monotone and strictly convex, then $x\succ y $ if and only if there is some $n,$ such that $x \succ^{R_n} y$.

	\end{theorem}

	\begin{proof}

	The ``if'' part is clear.
	
	\bigskip

	For the converse, given any $ y \in int({\mathbf X}), $ define
	\[
	T_y =\{ z \in {\mathbf X}: \mbox{There is no } n : y \succ^{R_n} z \},
	\]
	and define a new preference relation $ \succeq '$ by
	\[
	u {\succeq}' v \mbox{ if } \left\{ \begin{array} {ll}
	u \in \mbox{conv}(T_y \cup {\mathbf R}_+(v)) & \mbox{ if } v \notin T_y \\
	u \in T_y \cap {\mathbf R}_+(v) & \mbox{ if } v \in T_y 
	\end{array} \right.,
	\]
where $ {\mathbf R}_+(\cdot) $ denotes the upper contour set under the preference relation $ \succeq $.

	It can be verified that $ {\succeq}'$ is upper-semi continuous, monotone and convex, and we can define the demand correspondence  $ f^{\succeq'} $. We shall argue below that it  suffices to show that $ f^{\succeq}(\cdot)=f^{\succeq'}(\cdot) $. To prove this, note first that $f^{\succeq'}(\cdot) $ is non-empty for all $p$. Then suppose that for some $p \in \Rea^L_{++} $,  $ u \neq v= f^{\succeq}(p) $ but $ u \in f^{\succeq'}(p) $.  By the definition of $ {\succeq}', $ this can only be the case if  $ u \in T_y $ but $ v \notin T_y $. By the continuity of $ f^{\succeq}(\cdot)$, for sufficiently large number of observations, $n$, there must be a $j\in \{1,\ldots, n\}$ and a $p_j'$ sufficiently close to $p$ so that $ f^{\succeq} (p) \notin T_y $ but $ p_j' \cdot f^{\succeq} (p_j') > p_j' \cdot u $. This is a contradiction to transitivity since we would have that $ y $ is revealed preferred to $ f^{\succeq}(p_j') $ which is revealed preferred to $u$ and $ u $ cannot be in $ T_y $.

	Therefore, $ {\succeq} $ and $ {\succeq}' $ generate the same demand functions. By Theorem 2' in \cite{mascolell77-2}, this implies that the two preference relations coincide, which is only possible if $ T_y $ is equal to the upper contour set of $ {\succeq}$ at $y$. This completes the argument.

	\end{proof}

	The assumption that preferences are Lipschitzian cannot be dispensed with. This is surprising since with a finite number of observations one cannot test whether preferences are Lipschitzian. Nor can one test whether they are strictly convex. However, only the assumption of Lipschtitzian and strictly convex preferences guarantees that for sufficiently many observations the Afriat-inequalities no longer have a solution 
	
		Note also that in we assume prices become dense in all of $ \Rea^L_{++} $. For a given pair $ x,y \in {\mathbf X}$ of consumption bundles in the interior of the consumption set, and if indifference surfaces through interior bundles have closures in the interior, it suffices to observe prices in a compact set.

    \subsection{Identification of demand}
  	
	For any two bundles between which the individual is not indifferent, a finite number of random observations suffice to predict how the individual shall choose. In a market setting, it may be more relevant, however, to ask how the individual shall choose given arbitrary prices and incomes.
  
  	Analogously to the analysis above, we can define a revealed demand correspondence as
  	\[
	x^{R_n}(p) = \{ x \in {\mathbf X}: p \cdot x = 1, \mbox{ there is no } x' \in {\mathbf X}, p \cdot x' \le 1, x' \succ^{R_n} x \}. 
	\]  
 For each $p,$ the set $ x^{R_n}(p) $ can be defined by a finite number of linear inequalities that can be computed using Fourier-Motzkin elimination.
  
	It is clear that whenever preferences are convex and monotone,
	\[
	f^{\succeq}(p) \in x^{R_n}(p), \quad  n=1, \ldots.
	\]
 It is also clear that, from $n$ observations on demand, one cannot recover the demand function, and $ x^{R_n}(p) $ shall not be single valued. The following theorem shows that, for sufficiently large $n,$ demand can be arbitrarily well approximated by the revealed demand correspondence:

	\begin{theorem}
	\label{thm2}
	
	Suppose preferences are continuous, strictly convex, monotone and Lipschitzian.
Given any $ p \in \Rea^L_{++} $ and any $ \epsilon >, 0 $ there exists an $n$ such that
	\[
	x^{R_n}(p) \subset \{ x \in {\mathbf B}(p): \| f^{\succeq} (p) -x \| \le \epsilon \}.
	\]  
	
	\end{theorem}

	\begin{proof}

	For any $p \in \Rea^L_{++} $, since preferences are strictly convex, $ f^{\succeq}(p) \succ x, $ for all $ x \in {\mathbf B}(p), x \ne f^{\succeq}(p) $.By Theorem ~\ref{thm1}, there exist some $n$ so that $ f^{\succeq}(p) \succ^{R_n} x $. Take any compact set
	\[ 
	{\mathbf K} \subset {\mathbf B}(p) \setminus \{ f^{\succeq}(p) \}. 
	\]
Whenever $ f^{\succeq}(p) \succ^{R_n} x, $ there must be an open set around $x$ so the $ f(p) $ is revealed preferred to any point in this set. Therefore there is a collection of open sets covering $ {\mathbf K}. $ Since $ {\mathbf K} $ is compact, there exists a finite subcover. For each of the finitely many points defining the subcover, there is a finite $n,$ such that $ f(p) $ is revealed preferred to that point. This completes the argument.
	
	\end{proof}

	\section{Equilibrium}

 	So far, our analysis has focused on the classical demand setting. More generally, economic theory derives relationships between the fundamentals of the economy, some of which may not be observable, and observed individual or aggregate behavior or equilibrium prices. It is then of interest to ask whether what is observed can be used to deduce individual preferences. We tackle this issue in the classical setting of Walrasian equilibrium and assume that individual behavior,  individual demand in particular, is not observable. 
	
	Observations consist only of equilibrium prices and individual endowments or equivalently, as we showed in Section \ref{sec:1}, of prices, individual incomes and the resulting aggregate demand.
	
	Throughout the section we assume that there are $H$ individuals with preferences $ \succeq^h $ for each $ h=1,\ldots, H$. Analogously to our analysis above,  we consider a sequence of arbitrary prices and individual incomes. We assume that $ \{ (p_k,(w_k^h)_{h \in {\mathbf H}} ), k=1, \ldots  \} $ become dense in $\Rea^{L}_{++} \times {\mathbf W} $ as $ k \rightarrow \infty, $ and we define observations on aggregate demand as
	\[
	D_k= d^{\succeq ^{\mathbf H}}(p_k,w_k^1,\ldots, w_k^H), \quad k=1,\ldots. 
	\]
	
	Given $n$ observations, it is useful to define the set of consumption allocations that are consistent with the $n$ observations in the sense that they add up to aggregate demand and satisfy the strong axiom: 
	\[
{\mathbf C}_n  = \left\{ x \in \Rea^{nHL}_+ : 
\begin{array}{l}
\sum_{h \in {\mathbf H}} x^h_k = D^{\succeq^{\mathbf H}}(p_k,w^{\mathbf H}_k),  \quad k=1,\ldots, n, \\ \\
p \cdot x^h_k = w_k,  \quad k=1,\ldots, n, \\ \\
(x^h_k,p_k)_{k=1}^n \ \ \text{ satisfy SARP}, \quad h \in {\mathbf H} 
\end{array}
\right\}.
	\] 

	We denote the projection of $ {\mathbf C}_n $ onto the coordinates corresponding to the $k'$th observation by $ {\mathbf C}_{nk} \subset \Rea^{HL}_+ ,$  $k=1,\ldots, n.$ That is, 
$ {\mathbf C}_{nk} $ is the set of all consumption vectors across the $H$ individuals that are on the budget hyperplane and add up to aggregate consumption, and for which there are consumptions for all other observations $ k' \ne k $ that also add up to the relevant aggregate consumptions, are on relevant budget plains and, crucially, altogether satisfy the strong axiom of revealed preference.

	\begin{figure}[tb]
		
	\begin{center}
		\includegraphics[scale=0.25]{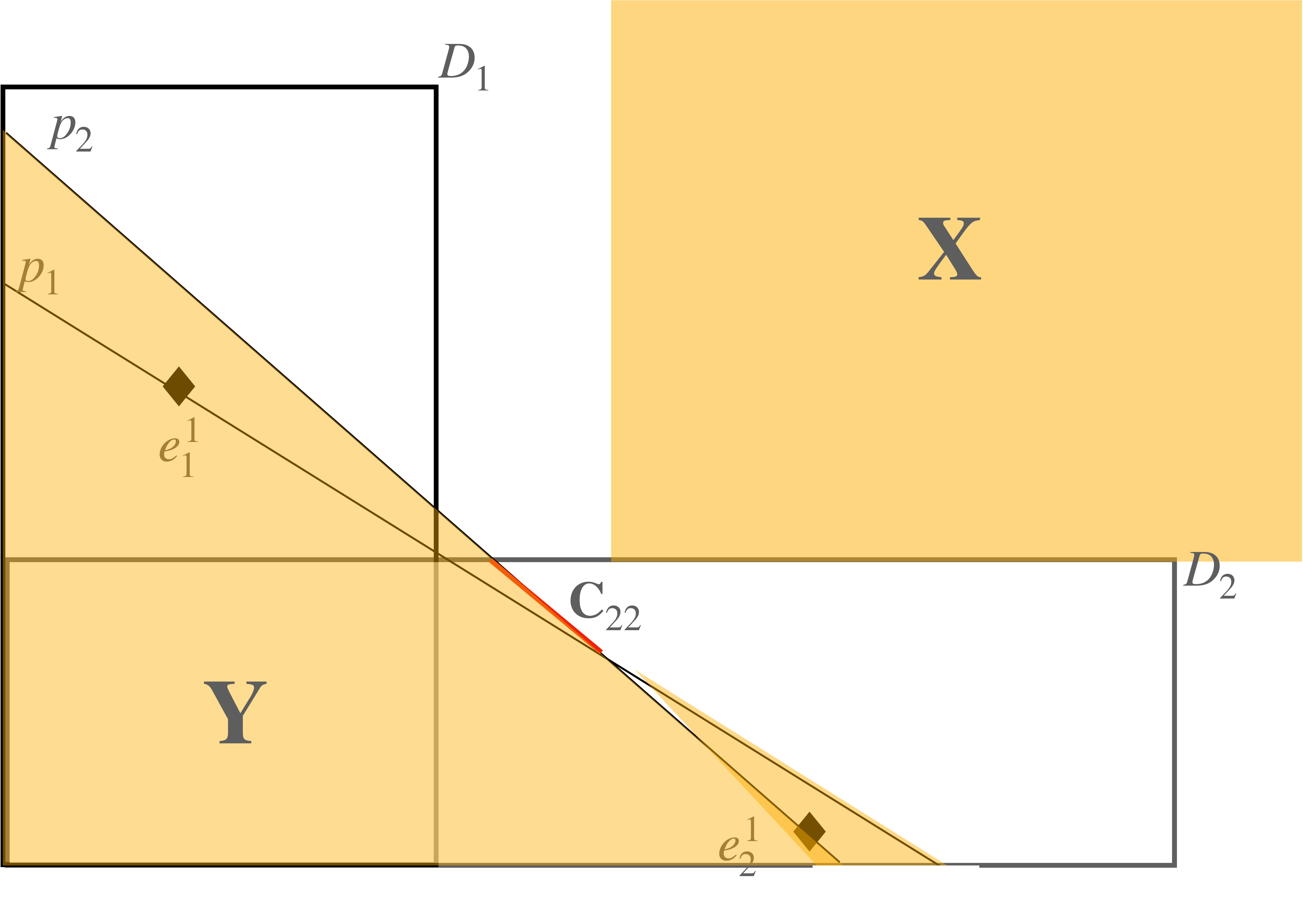}\\
	\end{center}
	\vspace{-0.5cm}
	\footnotesize{}.
	\caption{Revealed preference from pairwise choice}
	\label{fig:3}
	
	\end{figure}

	We illustrate how observations on aggregate demand, individual endowments and equilibrium prices limit the set of possible consumption vectors. Figure \ref{fig:3} shows two Edgeworth-boxes, $ D_1 $ corresponding to prices $p_1$ and endowments $ (e^1_1, D_1-e^1_1)$ and $ D_2 $ corresponding to prices $ p_2 $ and endowments $ (e^1_2, D_2-e^1_2) .$ The red line-segment on the budget line $p_2 $ corresponds to the possible consumption of individual 1 in the second observation. Since individual 1 must consume insight the first Edgeworth box given prices $ p_1, $ the only consumption consistent with the weak axiom lies on the red line-segment. Similarly to above we can then conclude that any bundle in the set $ {\mathbf X} $ is strictly preferred to any bundle in $ {\mathbf Y} $.

	Importantly, in our context, Assumption \ref{ass1} implies that, for sufficiently large $n,$ the sets $ {\mathbf C}_{nk} $ become {\it small} in the sense that they contain only a neighborhood around the profile of actual individual demands. Formally,

	\begin{lemma}
	\label{lemma1}
	
	Suppose preferences are continuous, monotone, strictly convex and Lipschitzian and suppose Assumption \ref{ass1} holds. Given any $\epsilon > 0$, any individual $h$ and any observation $k,$ there exists an $n$ such that 
	\[
	x^{\mathbf H} \in {\mathbf C}_{nk} 
	\Rightarrow  \| x^h- f^{\succeq^h}(\frac{p_k}{w^h_k}) \|  < \epsilon . 
	\]
	
	\end{lemma}

	\begin{proof}

	Suppose that, for some $k$ and some $h,$
$ y^h_k \ne f^{\succeq^h}(\frac{p_k}{w^h_k}) $ and that $ y_k^{\mathbf H} \in {\mathbf C}_{nk}  $ for all $n$. Then there is an infinite sequence of prices and choices $ (y^h_k,p_k) $, $ k=1,\ldots $ that satisfy SARP. \cite{mascolell78} implies that
 each $ y^h_k = f^{\succeq^h \prime}(\frac{p_k}{w^h_k}) $ for some $ \succeq^{h\prime }  \ne \succeq^h $. But, by Assumption 1 this implies that there must be some $j,$
 for which
 	\[
	d^{\succeq^{\mathbf H}}(p_j, w_j^{\mathbf H})
     \neq d^{\succeq^{\mathbf H}\prime}(p_j,  w_j^{\mathbf H}) .
     	\]
	This contradicts that $ y^{\mathbf H}_k \in {\mathbf C}_{nk} $ for all $n,$ and, as a consequence,  there must exist a $ \overline{n} $ such that $ y^{\mathbf H_k} \notin {\mathbf C}_{\overline n,k}  $.  

	For all $ n \le \overline{n}, $ if some $ (y^{\mathbf H}_k)_{k=1\ldots n-1} \in {\mathbf C}_{n-1} $  but there is no $ y^{\mathbf H}_{n} $ so that 
	\[
	(y^{\mathbf H}_k)_{k=1\ldots n} \in {\mathbf C}_{n}, 
	\] 
then for all $ y^{\mathbf H}_{n} $ with
	\[ 
	\sum_{h \in {\mathbf H}} y^h_n = D_{\overline n}
	\]
for some $h,$  there must be a budget feasible $ \tilde{y}^h $ that is strictly revealed preferred to $y^h$. Since this is a strict order, this must be true for an open neighborhood of $ y^{\mathbf H}_k $. Since the intersection of finitely many open neighborhood forms an open neighborhood there is an open neighborhood of $ y^{\mathbf H}_k $ which is not contained in $ {\mathbf C}_{\overline n,k}. $

The same argument (with possibly different $ \overline{n} $) applies to  all $ y^{\mathbf H}_k \in  \{ x \in {\mathbf B}(\frac{p_k}{w_k}): \| x- f^{\succeq^h}(\frac{p_k}{w_k}) \|  \ge  \epsilon\}. $

Since $ \{ x \in {\mathbf B}(\frac{p_k}{w_k}): \| x- f^{\succeq^h}(\frac{p_k}{w_k}) \|  \ge  \epsilon\} $ is a compact set, the open cover generated by all these $ y^{\mathbf H} $ has a finite sub-cover and hence there must be some finite $ n^* $ that gives the result.

	\end{proof}

	The result is surprising: Individual demands can be identified from the aggregate demand function! However, note that the result is consistent with the results in \cite{chiapporietal04} and in particular in \cite{matzkin06}. Observation of aggregate demand means that one observes the income effects of each individual's demand. If one assumes a rank condition (\cite{lewbel90}), as in \cite{chiapporietal04}, these income effects identify individual demands. The contribution of our lemma is to show that with finite data this is approximately true.

	\subsection{Identification of individual choice}
	
	As in Section \ref{sec:2} we want to develop a notion of `equilibrium revealed preferred' meaning that, through observations of equilibrium prices, it is revealed that an individual must prefer some bundle $x$ to another bundle $y$.
Given $ x,y \in {\mathbf X},$ we say that
$x$ is equilibrium revealed preferred to $y$ by individual $h$ given 
$n$ observations,
$ x \succ^{R^h_n} y $
if for all $ (x^{\mathbf H}_k)_{k=1,\ldots, n} \in {\mathbf C}_n $ and
  if there are $N$ observations indexed by $ i_1, \ldots, i_N \in \{1,\ldots,n \}$,
so that $x^h \ge x^h_{i_1}$,
and $ p_{i_j} \cdot x^h_{i_{j+1}} < p_{i_j} \cdot  x^h_{i_j} $ for all $ j=1,\ldots N$, as well as $ x^h_{i_N}\ge  y $.
While for the classical notion of revealed preferred, 
$x$ is said to be revealed preferred to $y$ if 
choices reveal that an individual must prefer $x$ to $y$, in this setting equilibrium prices reveal that an individual prefers $x$ to $y$. This is possible because incomes vary and Lemma \ref{lemma1} shows that individual choices can be recovered (approximately) from aggregate demand.

Obviously, if there is a $n$  such that $ x \succ^{R^h_n} y $ then $ x \succ^h y $. As in Section \ref{sec:2}, we show that the converse also holds, for sufficiently large $n$ $ x \succ^{R^h_n} y $ must hold whenever $ x \succ^h y $.

	\begin{theorem} \label{thm3}
Suppose individual preferences are continuous, monotone, strictly convex and Lipschitzian. Suppose they satisfy Assumption \ref{ass1}.
Given any $h \in {\mathbf H} $, any $ x,y \in {\mathbf X} $ with $ x \succ^h y$
there exists an $n$ such that
$$ x \succ^{R^h_n} y $$
	\end{theorem}
\begin{proof}
Theorem \ref{thm1} implies that there is an $\bar n$ such that for all $ n > \bar n$ there is a $ (x^{\mathbf H}_k)_{k=1,\ldots,n} \in {\mathbf C}_n $
so that the associated $ (x^h_k) $ to together with $ (p_k) $ imply that $x$ is revealed preferred to $y$ by individual $h$. 

Lemma \ref{lemma1} implies that for sufficiently large $n$ there can be no  other solutions $ (\tilde{x}^{\mathbf H}_k)_{k=1,\ldots,n} \in {\mathbf C}_n $ for which this does not hold.

\end{proof}

This is the main result of our paper. Sufficiently many observations on the equilibrium manifold allow us to infer how any one of the individuals in the economy will choose between two bundles.

It follows directly from the proof of Theorem of Theorem \ref{thm2} above that sufficiently many observations on the equilibrium manifold also allow us to predict each individual's Walrasian demand at arbitrary prices. Formally we can define the equilibrium revealed demand correspondence as
 $$ x^{R^h_n}(p) = \{ x \in {\mathbf X}: p \cdot x = 1, \mbox{ there is no } x' \in {\mathbf X}, p \cdot x' \le 1, x' \succ^{R^h_n} x \}. $$  
 
 Note that the definition is identical to the definition in Section \ref{sec:2} except that we replace ``revealed preferred'' by ``equilibrium revealed preferred.''
 As above, 
it is clear that whenever preferences are convex and monotone,
$$ f^{\succeq}(p) \in x^{R_n}(p), \quad \text{for} \ n=1, \ldots. $$
 
The following theorem shows that for sufficiently large $n$ demand can be arbitrarily well approximated by the revealed demand correspondence.

\begin{theorem}
\label{thm4}
Suppose preferences are continuous, strictly convex, monotone and Lipschitzian. Suppose also that Assumption \ref{ass1} holds.
Given any $ p \in \Rea^L_{++} $ and any $ \epsilon >, 0 $ there exists an $n$ such that
$$ 
 x^{R^h_n}(p) \subset
\{ x \in {\mathbf B}(p): \| f^{\succeq} (p) -x \| \le \epsilon \}.  $$
\end{theorem}
Given the result in Theorem \ref{thm4} and the proof of Theoerm \ref{thm3} the result follows immediately.

\subsection{Identification of Equilibrium}  

A more delicate issue is whether one can forecast equilibrium prices. Clearly a resolution of the transfer paradox would require this. An obvious problem that arises is that the possibility of multiple equilibria cannot easily be ruled out. But perhaps, even in the presence of multiplicity, one can predict one of the equilibria.
Unfortunately things become much more complicated because of the intricate relation between approximate equilibria and exact equilibrium. While it is true that
for any economy and any $ \delta> 0 $ there is  $ \epsilon  $
such that, if the norm of aggregate excess demand is smaller than $ \epsilon, $ prices are within $ \delta $ of equilibrium prices, \citet*{anderson1986almost}, there is no constructive algorithm to determine $ \epsilon $. In fact, \cite{richter1999non} seems to suggest that it is generally impossible to do so.

We will therefore focus on approximate equilibria.
Given $n$ observations, we define the  revealed approximate 
equilibrium correspondence as follows:
	\[
	{\mathbf P}^{R_n(\succeq^{\mathbf H})}(e^{\mathbf H}) = \left\{
	p \in \Delta^L :
	\begin{array}{l}
	\exists \ x^{\mathbf H}\in \Rea^{HL}_+,
	\begin{array}{l}
	p \cdot x^h=p\cdot e^h \mbox{ for all } h \in {\mathbf H}, \\ \\
	\sum_{h\in{\mathbf H}} (x^h-e^h)=0, 
	\end{array}
	\\ \\
	\not\exists \ y^{\mathbf H}\in \Rea^{HL}_+,
	\begin{array}{l}
	p \cdot y^h=p\cdot e^h \mbox{ for all } h \in {\mathbf H}, \\ \\ 
	y \succ^{R_n} x^h \mbox{ for all } h \in {\mathbf H} 
	\end{array}
	\end{array}
	\right\}.
	\]

	Note that this is a semi-algebraic set that can be written as the finite union and intersection of sets of the form $ \{x\in \Rea^{n}:g(x)> 0 \}$ or $\{ x \in \Rea^{n}: f(x)=0\}, $ where $f$ and $g$ are polynomials in $x$ with coefficients in $\Rea$. \citet*{basu2006algorithms}, quantifier-elimination can be used to compute the sets .

It is relatively way to show that, for any $ \epsilon > 0, $ these sets describe the set of all $ \epsilon $-equilibria.
Our final result is as follows:
\begin{theorem}
\label{thm5}
Suppose preferences are continuous, monotone, strictly convex and Lipschitzian, and Assumption \ref{ass1} holds.
For an $ \epsilon > 0, $ there exists an $n$ such that
$$ {\mathbf P}^{R_n(\succeq^{\mathbf H})}(e^{\mathbf H})  \subset 
 \{ p \in \Delta: \| D^{\succeq^h}(p,p\cdot e^1,\ldots, p \cdot e^H) -e\| < \epsilon \}. $$

\end{theorem}

\begin{proof}

	Define 
	\[
	{\mathbf D}^{R_n(\succeq^{\mathbf H})} =\left\{
D \in \Rea^L_+ : 
\begin{array}{l}
\exists \ x^{\mathbf H}:
\begin{array}{l}
\sum_{h\in{\mathbf H}} x^h=D, \\ \\
p \cdot x^h=p\cdot e^h \mbox{ for all } h \in {\mathbf H}, 
\end{array}
\\ \\
\not\exists \ y^{\mathbf H}: 
\begin{array}{l}
p \cdot y \le p \cdot e^h, \\ \\ 
y \succ^{R_n} x^h \mbox{ for all } h \in {\mathbf H} 
\end{array}
\end{array}
\right\}.
\]

It suffices to show that that, for each $ \epsilon, $ 
there is an $n$ such that 
$$ \sup_{d \in {\mathbf D}^{R_n}} \| d-D^{\succeq^h}(p) \| < \epsilon $$
The proof of Theorem 2 implies this result.

\end{proof}

As explained above, it seems difficult to go beyond this result and make statements about exact equilibria. This is due simply to the fact that small perturbations of fundamentals can have large effects on equilibrium prices. When it is known a priori that equilibrium is unique sufficiently many observations will allow us to identify the equilibrium prices within $ \epsilon $ since the set of approximate equilibrium prices must shrink eventually to consist only a neighborhood of the exact equilibrium.

	\bibliography{ifd-bibliography}
	
\end{document}